\documentclass{llncs}

\usepackage{amssymb,stmaryrd,mathrsfs,dsfont,mathtools}

\input xy
\xyoption{all}

\usepackage[ruled,linesnumbered]{algorithm2e}
\usepackage{algpseudocode}
\usepackage{tikz}
\tikzstyle{vertex}=[circle, draw, inner sep=0pt, minimum size=3pt, fill=black]
\newcommand{\vertex}{\node[vertex]}

\begin{document}

\frontmatter         

\pagestyle{headings}  

\mainmatter             

\title{Words for the Graphs with Permutation-Representation Number at most Three}

\titlerunning{Words for the Graphs with Permutation-Representation Number at most Three}

\author{Khyodeno Mozhui \and K. V. Krishna}

\authorrunning{Khyodeno Mozhui \and K. V. Krishna}

\institute{Indian Institute of Technology Guwahati, Assam, India,\\
	\email{k.mozhui@iitg.ac.in};
	\email{kvk@iitg.ac.in}}

\maketitle         

\begin{abstract}
	The graphs with permutation-representation number (\textit{prn}) at most two are known. While a characterization for the class of graphs with the \textit{prn} at most three is an open problem, we summarize the graphs of this class that are known so far. Although it is known that the \textit{prn} of trees is at most three, in this work, we devise a polynomial-time algorithm for obtaining a word representing a given tree permutationally. Consequently, we determine the words representing even cycles. Contributing to the class of graphs with the \textit{prn} at most three, we determine the \textit{prn} as well as the representation number of book graphs. 
\end{abstract}


\section{Introduction and Preliminaries}

Let $X$ be a finite set. A word over $X$ is a finite sequence of elements of $X$ written by juxtaposing them. A word $u$ is called a subword of a word $w$, denoted by $u \le w$, if  $u$ is a subsequence of $w$. Suppose $w$ is a word over $X$, and $Y$ is a subset of $X$. The subword (possibly the empty sequence) of $w$ obtained by deleting all letters belonging to $X \setminus Y$ from $w$ is denoted by $w_Y$. Suppose $w$ is a word over $X$ and $a, b \in X$. We say $a$ and $b$ alternate in $w$ if $w_{\{a,b\}}$ is of the form $abababa\cdots$ or $bababa\cdots$ of even or odd length. A word in which every letter appears exactly $k$ number of times is called a $k$-uniform word. 

A simple graph $G = (V, E)$ is called a word-representable graph if there exists a word $w$ over its vertex set $V$ such that, for all $a, b \in V$, $\overline{ab} \in E$ if and only if $a$ and $b$ alternate in $w$. In \cite{kitaev08}, it was proved that if a graph is word-representable then there are infinitely many words representing it.  Further, every word-representable graph is represented by a $k$-uniform word, for some positive integer $k$. For a detailed introduction to the theory of word-representable graphs, one may refer to the monograph by Kitaev and Lozin \cite{kitaev15mono}.

A word-representable graph is said to be $k$-word-representable if it is represented by a $k$-uniform word. The smallest $k$ such that a graph $G$ is $k$-word-representable is called the representation number of the graph, and it is denoted by $\mathcal{R}(G)$. The class of word-representable graphs with representation number $k$ is denoted by $\mathcal{R}_k$ and the class with representation number at most $k$ is denoted by $\mathcal{R}_{\le k}$. While $\mathcal{R}_1$ is the class of complete graphs, it was shown in \cite{halldorsson11} that $\mathcal{R}_{\le 2}$ is the class of circle graphs. A circle graph is the intersection graph of a set of chords of a circle. The class of circle graphs was characterized by three minimal forbidden vertex minors \cite{bouchet}. In \cite{halldorsson11}, it was proved that deciding whether a word-representable graph has representation number $k$, for any fixed $k$, $  3 \le k \le \lceil \frac{n}{2} \rceil$, is NP-complete. Nevertheless, the representation number for some specific classes of graphs was obtained, in addition to some isolated examples. For example, while the representation number was established for prisms \cite{kitaev13} and crown graphs \cite{glen18}, upper bounds for the representation number were obtained for 3-colorable graphs \cite{halldorsson16}, bipartite graphs \cite{KM_KVK_1} and $k$-cubes \cite{broere_2019}.

A permutationally representable graph is a word-representable graph that can be represented by a word of the form $p_1p_2\cdots p_k$ where each $p_i$ is a permutation of its vertices. Considering the number of permutations $k$, such a graph is called a permutationally $k$-representable graph. It was shown in \cite{kitaev08order} that the class of permutationally representable graphs is precisely the class of comparability graphs - the graphs which admit transitive orientation. An orientation of a graph is an assignment of direction to each edge so that the graph obtained is a directed graph. If the adjacency relation on the vertices of the resulting directed graph is transitive then the orientation is said to be transitive. The permutation-representation number (in short, \textit{prn}) of a comparability graph $G$, denoted by $\mathcal{R}^p(G)$, is the smallest number $k$ such that the graph is permutationally $k$-representable. The class of graphs with \textit{prn} $k$ is denoted by $\mathcal{R}^p_k$, and the class with the \textit{prn} at most $k$ is denoted by $\mathcal{R}^p_{\le k}$. Note that $\mathcal{R}(G) \le \mathcal{R}^p(G)$ and $\mathcal{R}_{\le k} \subseteq \mathcal{R}^p_{\le k}$. In \cite{halldorsson11}, it was shown that determining the \textit{prn} of a comparability graph is same as determining the dimension of the partially ordered set (poset) induced by the comparability graph. In \cite{yanna82}, Yannakakis showed that deciding whether a poset with $n$ elements has dimension at most $k$, for fixed $k$, where $  3 \le k \le \lceil \frac{n}{2} \rceil$, is NP-hard. Accordingly, it is NP-hard to decide whether a comparability graph with $n$ vetices has the \textit{prn} at most $k$, for fixed $k$, where $  3 \le k \le \lceil \frac{n}{2} \rceil$. In \cite{5}, it was shown that it is NP-hard to approximate the dimension of a poset to within almost a square root factor.

It is clear that $\mathcal{R}_1 = \mathcal{R}^p_1$ is the class of complete graphs. In this paper, we reconcile the class $\mathcal{R}^p_{\le 2}$ as the class of permutation graphs. For $k \ge 3$, characterizations for $\mathcal{R}_k$ and for $\mathcal{R}^p_k$ are open problems. Focusing on identifying the subclasses of graphs with \textit{prn} three, first we recall the classes of comparability graphs with \textit{prn} at most three. Although it was known that the trees are in the class $\mathcal{R}^p_{\le 3}$, explicit construction of words representing trees permutationally is not available. In this work, we give a polynomial-time algorithm to generate 3-uniform words that represents a tree permutationally. This algorithm enables us to generate words for paths and even cycles\footnote{The cycles on an odd number of vertices are not comparability graphs.} as well. Further, using the algorithm, we establish the \textit{prn} of book graphs. We also determine the representation number of book graphs and show that it is a class of graphs for which the \textit{prn} and  representation number are the same.

\section{The Class $\mathcal{R}^p_{\le 2}$} \label{permutation_graph}

In this section, we reconcile the characterization of $\mathcal{R}^p_{\le 2}$ in terms of permutation graphs. In \cite{Even_1972,Punueli_1971}, the concept of a permutation graph was introduced  and characterized as a comparability graph whose complement is also a comparability graph.    

A graph $G = (V, E)$ is called a permutation graph if there exist two permutations $p_1$ and $p_2$ on $V$ such that two vertices $a, b \in V$ are adjacent if and only if $ab \le p_1$ and $ba \le p_2$. It can be easily observed that permutation graphs are permutationally 2-representable graphs. For instance, the word $p_1r(p_2)$, where $r(p_2)$ is the reversal of $p_2$, represents permutationally a permutation graph. Conversely, if a graph is permutationally 2-representable, it is a permutation graph. Note that a complete graph is a permutation graph. Accordingly, we record the characterization of the class $\mathcal{R}^p_{\le 2}$ in the following result.

\begin{theorem}	
	The class of permutation graphs is precisely the class of graphs with \textit{prn} at most two.	
\end{theorem}

While \cite{Gallai} provides a list of forbidden induced subgraphs for permutation graphs, a succinct characterization of permutation graphs was given in terms of forbidden Seidel minors in \cite{Vincent}. In \cite{golumbic}, an $O(n^3)$-time procedure was given for constructing words representing any permutation graph permutationally. 

\section{The Class $\mathcal{R}^p_{\le 3}$}

This section, summarizes the comparability graphs whose \textit{prn} was known to be at most three. Clearly, the class of permutation graphs belongs to this class. A treelike comparability graph is a graph that admits a transitive orientation in which the cover graph of the associated poset is a tree. In \cite{trotter_dim}, it was proved that the treelike comparability graphs have the \textit{prn} at most three. Consequently, the trees have the \textit{prn} at most three. An outerplanar graph is a planar graph in which all its vertices lie on the unbounded face of a planar drawing. The bipartite graphs that are outerplanar are 2-word-representable and have the \textit{prn} at most three \cite{kitaev08,Felsner_2015}. It was shown in \cite{Guenver_2004} that the split comparability graphs are in $\mathcal{R}^p_{\le 3}$, where a split graph is a graph whose vertices can be partitioned into a clique and an independent set. Recently, it has been proved that the extended crown graphs defined by posets of width two are also in $\mathcal{R}^p_{\le 3}$ \cite{KM_KVK_2}.

\section{Trees}

In this section, through a polynomial-time procedure, we construct three permutations on the vertices of a tree whose concatenation represents the tree permutationally. 

Let $T$ be a tree, i.e., a connected acyclic graph, and $r$ be the root of $T$, i.e., a distinguished vertex of $T$. A vertex $a$ is called a child of a vertex $b$ if $b$ is adjacent to $a$ on the path from the root $r$ to $a$. In this case, $b$ is called the parent of $a$. Any other vertex $c$ (other than the parent $b$) on the path from $r$ to $a$ is called an ancestor of $a$. If a vertex has no children, then it is called a leaf.
For a positive integer $k$, a $k$-ary tree is a tree in which every vertex has at most $k$ children.
	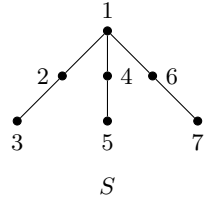
\begin{figure}
	\centering
	\begin{tikzpicture}[scale=0.6]
		\node (A) at (0,-4) [label=above:$S$] {};  
		\vertex (1) at (0,0) [label=above:$1$] {};  
		\vertex (2) at (-1,-1) [label=left:$2$] {}; 
		\vertex (3) at (0,-1) [label=right:$4$] {}; 
		\vertex (4) at (1,-1) [label=right:$6$] {}; 
		\vertex (5) at (-2,-2) [label=below:$3$] {}; 
		\vertex (6) at (0,-2) [label=below:$5$] {}; 
		\vertex (7) at (2,-2) [label=below:$7$] {}; 
		\path
		(1) edge (2)
		(1) edge (3)
		(1) edge (4)
		(2) edge (5)
		(3) edge (6)
		(4) edge (7);
	\end{tikzpicture}
	\caption{Forbidden induced subgraph for the class $\mathcal{R}^p_2$}
	\label{S}
\end{figure}

In view of the characterization of permutation graphs \cite[Theorem 10]{Vincent}, we have the following remark. 
\begin{remark}\label{thm-10vincent}
	The trees which are not having the graph $S$ given in Figure \ref{S} as an induced subgraph are permutation graphs. Hence, such trees on at least three vertices have the \textit{prn} two. 
\end{remark}	
 
Accordingly, $\mathcal{R}^p(S) \ge 3$. Further, it can be observed that $\mathcal{R}^p(S) = 3$, as the word  234615767452132345671 represents the graph $S$ permutationally. In the following, we show that the trees containing $S$ as a induced subgraph have the \textit{prn} three.

Let $T$ be an arbitrary tree. Suppose $k$ is the maximum degree of $T$. We can view $T$ as a $k$-ary tree. First, assign an odd number $a_r$ as the label of the root. Then, set all the children of an odd label vertex with even numbers as labels and vice versa. Let $A$ be the set of vertices with even labels and $B$ be the set of vertices with odd labels in $T$. Note that $\{A, B\}$ is a bipartition of $T$. With this description, a tree $T$ is always identified by $(A \cup B, E)$ in the following.  

\begin{algorithm}[!htb] 
	\caption{Generating permutations for a $k$-ary tree.}
	\label{algo-1}
	\KwIn{A $k$-ary tree $T = (A \cup B, E)$ with the root $a_r$.}
	\KwOut{Three permutations $p_1$, $p_2$ and $p_3$ of the vertices of $T$.}
	Let $Q$ be a queue.\\
	Append $a_r$ to $Q$.\\
	Initialize $p_1 = a_r ,p_2 = a_r, p_3 = a_r$;\\
	\While{$Q$ is not empty}{
		Remove the the first element $a$ from $Q$. \\
		\If{$a$ is not a leaf}{
			Let $a_{1} < a_{2} < \dots < a_{k}$ be the children of $a$.\\
			\If{$a$ is odd}{Replace $a$ by the word $a_{1} a_{2} \cdots a_{k} a$ in $p_1$.\\
				Replace $a$ by the word $a_{k} a_{{k-1}} \cdots a_{1} a$ in $p_2$.\\
				$p_3 \gets a_{1} a_{2} \cdots a_{k} p_3$}
			\Else{
				Let $a'$ be the parent of $a$ and $b_{1} < b_{2} < \dots < b_{k}$ be the children of $a'$.\\
				\If{$a = b_1$}{
					Replace $a$ by the word $aa_{1} a_{2} \cdots a_{k}$ in $p_1$.\\
					Replace $aa'$ by the word $a a' a_{k} a_{{k-1}} \cdots a_{1}$ in $p_2$.\\
					Replace $a$ by the word $a a_{1} a_{2} \cdots a_{k} $ in $p_3$.}
				\Else{ Suppose $a = b_l$ ($1 < l \le k$). \\
					Let $D_{b_i}$ denote the permutation consisting only the children of $b_i$ written in the increasing order of their indices.\\
					Replace $a b_{l+1} b_{l+2} \cdots b_{k} a' D_{b_2} D_{b_3} \cdots D_{b_{l-1}}$ by the word $a  b_{l+1} b_{l+2} \cdots b_{k} a' D_{b_2} D_{b_3} \cdots D_{b_{l-1}} a_{1} a_{2} \cdots a_{k}$ in $p_1$. \\
					Replace $a$ by the word $a a_{k} a_{{k-1}} \cdots a_{1}$ in $p_2$.\\
					Replace $a$ by the word $a a_{1} a_{2} \cdots a_{k} $ in $p_3$. 			}}
			Append the children of $a$ which are not leaves to $Q$ in the increasing order of their labels.
		}
		
	}
\Return $p_1,p_2,p_3$
\end{algorithm}

Given a $k$-ary tree $T = (A \cup B, E)$ with the root $a_r$, Algorithm \ref{algo-1} produces three permutations $p_1,p_2$ and $p_3$ on the vertices of $T$. Note that the permutation $p_i$ ($1 \le i \le 3$) obtained in each step of the algorithm is a subword of the respective permutation updated in the next step. In Theorem \ref{complete}, we prove the correctness of Algorithm \ref{algo-1} that the word $p_1p_2p_3$ represents $T$. A demonstration of Algorithm  \ref{algo-1} for obtaining a word that represents a $k$-ary tree permutationally is given in Appendix \ref{tree-demo}.

Note that Algorithm \ref{algo-1} follows breadth first search (BFS) for visiting the vertices of the input tree. If $n$ is the number of vertices, visiting the vertices takes $O(n)$ time, as input being a tree \cite{cormen}. When each vertex is visited, the permutations $p_1, p_2$, and $p_3$ are updated, which takes $O(n)$ time at each vertex. Hence, Algorithm \ref{algo-1} runs in polynomial time, as stated in the following result. 
\begin{theorem}\label{complexity-1}
	Algorithm \ref{algo-1} runs in $O(n^2)$ time on any tree $T$ with $n$ vertices.  
\end{theorem}

\begin{lemma}
	Let $T=(A \cup B, E)$ be a $k$-ary tree. Suppose $a$ is an ancestor of $b $ in $T$. 
	\begin{enumerate}
		\item If $a \in A$ then $ab \le p_1$.
		\item If $a \in B$ then $ab \le p_1$ or $ab \le p_2$.
	\end{enumerate}
\end{lemma}

\begin{proof}
		Let $a=a_0, a_1, a_2, \ldots, a_k, a_{k+1}=b$ be the path from $a$ to $b$ such that $a_i$ is the parent of $a_{i+1}$ for all $0 \le i \le k$.

Note that an ancestor of a vertex does not include its parent.
As per the algorithm, the ancestor of any vertex is visited first. Therefore, $a$ is visited first then $a_1, a_2, \ldots, a_k$ and lastly $b$.
\begin{enumerate}
	\item If $a \in A$ then $a_1 \in B$. We have $a a_1 \le p_1$. Next $a_1$ is visited and $a_2 \in A$ being the child of $a_1$, we have $a a_2 \le a a_2 a_1 \le p_1$. Then $a_2$ is visited and $a_3 \in B $ being the child of $a_2$. If $a_2$ is the leftmost child of $a_1$ then $a a_3 \le a a_2 a_3 a_1 \le p_1$. If $a_2$ is not the leftmost child of $a_1$, then $a a_3 \le a a_2 a_1 a_3 \le p_1$. Continue this till $a_k$ is visited. If $a_k$ is odd, then $a b \le a b a_k \le p_1$. If $a_k$ is even, then $ab \le a a_k b \le p_1$.
	Hence, $ab \le p_1$.
	
	\item If $a \in B$ then $a_1 \in A$. We have $a_1 a \le p_1$ and $a_1 a \le p_2$. Next, $a_1$ is visited, and $a_2 \in B$ being the child of $a_1$. If $a_1$ is the leftmost child of $a$, then $a a_2 \le a_1 a a_2 \le p_2$. If $a_1$ is not the leftmost child of $a$, then $a a_2 \le a_1 a a_2 \le p_1$. Then $a_2$ is visited and $a_3 \in A$ being the child of $a_2$, we have $a a_3 \le a a_3 a_2 \le p_2$ (in case of $a_1$ being the leftmost child of $a$) or $a a_3 \le a a_3 a_2 \le p_1$ (in case of $a_1$ not being the leftmost child of $a$). Continue this till $a_k$ is visited; if $a_k$ is odd, then $a b \le a b a_k \le p_1$ (in case of $a_1$ not being the leftmost child of $a$) or $a b \le a b a_k \le p_2 $ (in case of $a_1$ being the leftmost child of $a$). If $a_k$ is even, then $a b \le a a_k b\le p_1$ or $ a a_k b \le p_2$.
	Hence, $ab \le p_1$ or $ab \le p_2$. \qed
\end{enumerate}	
\end{proof}

\begin{lemma}
	If $a$ is an ancestor of $b$ in a $k$-ary tree $T=(A \cup B, E)$ then $ba \le p_3$. 
\end{lemma}

\begin{proof}
	Let $a=a_0, a_1, a_2, \ldots, a_k, a_{k+1}=b$ be the path from $a$ to $b$ such that $a_i$ is the parent of $a_{i+1}$, for all $0 \le i \le k$.

As per the algorithm, the ancestor of any vertex is visited first. Therefore, $a$ is visited first then $a_1, a_2, \ldots, a_k$ and lastly $b$. 
\begin{itemize}
	\item If $a \in A$ then $a_1 \in B$. We have $a a_1 \le p_3$. Next $a_1$ is visited and $a_2 \in A$ being the child of $a_1$, we have $a_2 a \le a_2 a a_1 \le p_3$. Then $a_2$ is visited, and $a_3 \in B $ being the child of $a_2$, we have $a_3 a \le a_2 a_3 a \le p_3$. Continue this till $a_k$ is visited. If $a_k$ is odd, then $ba \le b a_k a \le p_3$. If $a_k$ is even, then $ba \le  a_k b a \le p_3$.
	Hence, $ba \le p_3$.
	\item If $a \in B$ then $a_1 \in A$. We have $a_1 a \le p_3$. Next $a_1$ is visited and $a_2 \in B$ being the child of $a_1$, we have $a_2 a \le a_1 a_2 a \le p_3$. Then $a_2$ is visited and $a_3 \in A$ being the child of $a_2$, we have $a_3 a \le a_3 a_2 a \le p_3$ Continue this till $a_k$ is visited if $a_k$ is odd, then $ba \le b  a_k a \le p_3$. If $a_k$ is even, then $b a \le a_k b a\le p_3$.
	Hence, $ba \le p_3$. \qed
\end{itemize}		
\end{proof}

\begin{corollary}\label{ancestor}
	Let $T=(A \cup B, E)$ be a $k$-ary tree. If $a$ is an ancestor of $b$ in $T$ then $ab \le p_i$ and $ba \le p_j$, for some $i \ne j$. In otherwords, $a$ and $b$ do not alternate in the word $p_1p_2p_3$.
\end{corollary}

\begin{lemma}\label{common_ancestor}
	Let $T=(A \cup B, E)$ be a $k$-ary tree. If neither $a$ is an ancestor of $b$ nor $b$ is an ancestor of $a$ in $T$, then $a b \le p_1$ and $b a\le p_2$. In otherwords, $a$ and $b$ do not alternate in the word $p_1p_2p_3$.
\end{lemma}

\begin{proof}
	
	Note that the set $A$ denotes the vertices of $T$ with an even label, and the set $B$ denotes the vertices of $T$ with an odd label vertex.
	
	Suppose $a$ and $b$ are siblings, where $a < b$ and $c$ is the parent. As per Algorithm \ref{algo-1}, $c$ is visited first, and as such, we get $a b \le p_1 $ and $b a \le p_2$ whenever $c$ is of odd or even label.
	
	Assume $a$ and $b$ are not siblings. The root $a_r$ is a common ancestor of $a$ and $b$. However, consider $c$ the closest common ancestor of $a$ and $b$. 
	
	For $ 1\le i \le k_1-1$, let $a_i$ be parent of $a_{i+1}$ and let $c$ be parent of $a_1$ and $a_{k_1}$ be parent of $a$. So also, for $ 1\le i \le k_2-1$, let $b_i$ be parent of $b_{i+1}$ and let $c$ be parent of $b_1$ and $b_{k_2}$ be parent of $b$. Note that $a_1 < b_1$.
	
	Here, $c$ is visited first as per the algorithm.
	
	\begin{itemize}
		\item Suppose $c \in A$.  Here, $a_1, b_1 \in B$ and we have $a_1 b_1 \le c a_1 b_1 \le p_1$ and $b_1 a_1 \le c b_1 a_1 \le p_2$. Since $a_1 < b_1$, $a_1$ is visited first, and we get $ a_2 b_1 \le a_2 a_1 b_1 \le p_1$ and $b_1 a_2 \le b_1 a_2 a_1 \le p_2$. Then $b_1$ is visited, we have $ a_2 b_2 \le a_2 b_2 b_1  \le p_1$ and $b_2 a_2 \le b_2 b_1 a_2 \le p_2$.
		
		\item Suppose $c \in B$. If $a_1$ is the leftmost child of $c$, then we have $a_1 b_1\le a_1b_1c \le p_1$ and $b_1 a_1 \le b_1a_1 \le p_2$.  Since $a_1 < b_1$, $a_1$ is visited first, and we get $ a_2 b_1 \le a_1 a_2 b_1 c \le p_1$ and $b_1 a_2 \le b_1 a_1 c a_2 \le p_2$. Then $b_1$ is visited, for which we get $ a_2 b_2 \le a_1 a_2 b_1 c b_2  \le p_1$ and $b_2 a_2 \le b_1 b_2 a_1 c a_2 \le p_2$.
		
		If $a_1$ is not the leftmost child of $c$, then we have $a_1b_1 \le a_1b_1c \le p_1$ and $b_1a_1 \le b_1a_1c\le p_2$. As $a_1<b_1$, $a_1$ is visited first, we get $b_1a_2 \le a_1b_1ca_2 \le p_1$ and $b_1a_1a_2c \le p_2$. Then $b_1$ is visited and we get $a_2 b_2 \le a_1 b_1 c a_2 b_2 \le p_1$ and $b_2a_2 \le b_1 b_2 a_1 a_2 c \le p_2$.
	\end{itemize}
	
	Subsequently, in the above two cases, we arrive at a common conclusion that is $a_k b_k \le p_1$ and $b_k a_k \le p_2$, where $k = \min \{k_1,k_2\}$. Note that $a_k$ and $b_k$ are at equal height from $c$. Also, $a = a_{k_1 +1}$ and $b = b_{k_2+1}$.
	\begin{itemize}
		\item Suppose $k = k_1 = k_2$, we have $ab \le p_1$ and $ba \le p_2$.
		\item Suppose $k_1 < k_2$. We have $a_{k_1+1}b_{k_1+1} \le p_1$ and $b_{k_1+1}a_{k_1+1} \le p_2$. Since $a$ is not a parent of any descendants of $b_{k_1+1}$, all descendants of $b_{k_1+1}$ will occur on the right side of $a$ in $p_1$ and occur on the left side of $a$ in $p_2$. Therefore, $a b \le p_1$ and $b a \le p_2$.
		\item Suppose $k_2 < k_1$. We have $a_{k_2 +1} b_{k_2+1} \le p_1$ and $b_{k_2 +1} a_{k_2+1} \le p_2$. Since $b_{k_2 +1}$ is not a parent to any descendants of  $a_{k_2+1}$, all descendants of $a_{k_2+1}$ will occur on the left side of $b$ in $p_1$ and will occur on the right side of $b$ in $p_2$. Thus, $ab \le p_1$ and $ba \le p_2$. 
	\end{itemize}

	Hence, $ab \le p_2 $ and $ba \le p_3$. \qed

\end{proof}

\begin{theorem}\label{complete}
	The word $p_1p_2p_3$ obtained by Algorithm 1 represents a $k$-ary tree $T=(A \cup B, E)$ permutationally.
\end{theorem}

\begin{proof}
	We claim that for all $a\in A$ and $b \in B$, $\overline{a b} \in E$ if and only if $a$ and $b$ alternate in $p_1p_2p_3$, i.e., $ a b \le p_i$, for all $i=1,2,3$.
	
	Suppose $a$ and $b$ are adjacent in $T$. We prove this part in the following two cases:
	\begin{itemize}
		\item Case 1: $a$ is a parent of $b$. As per the algorithm, $a$ is visited first and so $a b \le p_1$, $a b \le p_2$ and $a b \le p_3$. Once $a$ and $b$ are inserted in the new updated word, $a b $ is always a subword of the final permutations $p_1,p_2$ and $p_3$.
		\item Case 2: $b$ is a parent of $a$. As per the algorithm, $b$ is visited first and we get $a b \le p_1$ , $a b \le p_2$ and $a b \le p_3$. Once $a$ and $b$ are inserted in the new updated word, $ a b$ is always a subword of the final permutations $p_1,p_2$ and $p_3$.
	\end{itemize}
	Hence, $a b \le p_i$, for all $i=1,2,3$.
	
	Conversely, suppose $a$ and $b$ are not adjacent in $T$. We need to show that there exist two distinct permutations such that $ab \le p_i$ and $ba \le p_j$, for some $i$ and $j$. We deal with this part in the following two cases:
	\begin{itemize}
		\item Case 1: Without loss of generality, assume that $a$ is an ancestor (not parent) of $b$ then from Corollary \ref{ancestor}, we have $ab \le p_i$ and $ba \le p_j$, for some $i$ and $j$. 
		\item Case 2: We know that $a_r$ is a common ancestor of $a$ and $b$. Let $c$ be their closest common ancestor. By Lemma \ref{common_ancestor}, we have that $ab \le p_i$ and $ba \le p_j$, for some $i$ and $j$.
	\end{itemize}
	Hence, $p_1p_2p_3$ represents a $k$-ary tree $T$ permutationally. \qed
\end{proof}

\begin{corollary}
	The \textit{prn} of trees which are having $S$ as the induced subgraph is three.
\end{corollary}

In the following result we summarize the \textit{prn} of the class of trees.

\begin{theorem}
	Let $T$ be a tree on $n$ vertices.
	\begin{enumerate}
		\item If $n =1$ or $2$ then $\mathcal{R}^p(T) = 1$.
		\item For $n \ge 3$, if $T$ does not contain $S$ as an induced subgraph then $\mathcal{R}^p(T) = 2$.
		\item For $n \ge 3$, if $T$ contains $S$ as an induced subgraph then $\mathcal{R}^p(T) = 3$.   
	\end{enumerate}
\end{theorem}

\section{Even Cycles}

In this section, we construct three permutations on the vertices of a cycle $C_n$ for an even number $n \ge 6$ and show that $\mathcal{R}^p(C_n) = 3$. In this connection, first we construct two permutations of a special type (cf. Remark \ref{unique}) whose concatenation represents a path permutationally. These permutations will be extended to give the desired permutations for an even cycle $C_n$.

Let $P_n$ be the path on $n$ vertices. For $n = 1$ and $2$, being complete graphs, $\mathcal{R}^p(P_n) = 1$. In view of Remark \ref{thm-10vincent}, all paths are permutation graphs. Hence, for $n \ge 3$,  $\mathcal{R}^p(P_n) = 2$. In the following, we construct words consisting of two permutations representing $P_n$ permutationally. 
First, we state the following remark from a result in \cite{KM_KVK_2}. 
\begin{remark}\label{letter_right}
	Let $G =( A \cup B, E)$ be a bipartite graph and $p$ be a permutation occurring in a word representing $G$ permutationally. If any vertex adjacent to a vertex $a \in A$ in $G$ appears on the right side (left side) of $a$ in $p$, then every other	vertex adjacent to $a$ in $G$ also appears on the right side (left side) of $a$ in $p$.	
\end{remark}

For $n \ge 3$, suppose $\{a_1,a_2,\ldots,a_n\}$ is the vertex set of $P_n$ where $a_i$ is adjacent to $a_{i+1}$, for $1 \le i \le n-1$. Note that $\{A, B\}$, where $A=\{a_i: i \ \text{is even}\}$ and $B =\{a_j: j \ \text{is odd}\}$, is the bipartition of $P_n$. 	For $a_i \in A$ and $a_j \in B$, if $a_i$ and $a_j$ are adjacent, using Remark \ref{letter_right}, assume that $a_ia_j \le p_i$ in any permutation $p_i$ occurring in a word representing $P_n$ permutationally. As $P_n$ is not a complete graph, any word representing $P_n$ must contain at least two permutations. On the input $P_n$, Algorithm \ref{algo-2} produces two permutations on the vertices of $P_n$. 

\begin{algorithm}[!htb]
	\caption{Generating permutations for a path.}
	\label{algo-2}
	\KwIn{A path $P_n$.}
	\KwOut{Two permutations on the vertices of $P_n$.}
	Run Algorithm \ref{algo-1} on $P_n$ with $a_1$ as the root. \\
	\Return $p_2, p_3$. 
\end{algorithm}

We now prove the correctness of Algorithm \ref{algo-2}.

\begin{lemma}\label{path}
	The word $p_2p_3$ represents $P_n$ permutationally.
\end{lemma}

\begin{proof}	
	Depending on whether $n$ is odd or even, in two cases, we show that $p_2p_3$ represents $P_n$ permutationally. It is evident from Algorithm \ref{algo-2} that the permutation are as follows:
	\begin{align*}
	p_2 &= \left\{\begin{array}{ll}
		a_2 a_1 a_4 a_3 a_6 a_5\cdots a_{n-1}a_{n-2}a_n, & \text{if $n$ is odd;}\\
		a_2 a_1 a_4 a_3a_6 a_5 \ldots a_{n-2}a_{n-3}a_n a_{n-1}, & \text{if $n$ is even.}
	\end{array} \right.\\
	p_3 &= \left\{\begin{array}{ll}
		a_{n-1} a_n a_{n-3} a_{n-2} \cdots a_6 a_7 a_4 a_5 a_2 a_3 a_1, & \text{if $n$ is odd;}\\
		a_n a_{n-2} a_{n-1} a_{n-4} a_{n-3} \ldots a_6 a_7 a_4 a_5 a_2 a_3 a_1, & \text{if $n$ is even.}
	\end{array} \right.
	\end{align*}
	It can be verified that the word $p_2p_3$ represents $P_n$ permutationally. The details are given Appendix \ref{app_lem4}.\qed
 \end{proof}

\begin{remark}\label{unique}
	In view of Remark \ref{letter_right}, the permutations $p_2$ and $p_3$ constituting a word representing $P_n$ permutationally are unique except for the ordering of even and odd indices.
\end{remark}

For $n \ge 3$, let $C_n$ be the cycle on $n$ vertices, say $\{a_1, \ldots, a_n\}$ such that $a_i$ is adjacent to $a_{i+1}$ (for $1 \le i < n$) and $a_n$ is adjacent to $a_1$. Note that $C_n = P_n \cup \{\overline{a_na_1}\}$ and the notation given for $P_n$ will be followed. If $n$ is even then $C_n$ is called an even cycle. While $C_4$ is a permutation graph, even cycles on at least six vertices are not permutation graphs \cite{Gallai}. For $n \ge 6$, we give Algorithm \ref{algo-3} to produce three permutations on the vertices of $C_n$ whose concatenation represents $C_n$.

\begin{algorithm}[!htb]
	\caption{Generating permutations for an even cycle.}
	\label{algo-3}
	
	\KwIn{An even cycle $C_n$.}
	\KwOut{Three permutations $p_1, p_2, p_3$ of the vertices of $C_n$.}
	Let $p_A$ be a permutation on the vertices of $A \setminus \{a_n\}$.
	
	Let $p_B$ be a permutation on the vertices of $B \setminus \{a_1, a_{n-1}\}$.
	
	Run Algorithm \ref{algo-2} on the path $P_{n-1}$.
	\Comment{$p_2$ and $p_3$ are returned.}
	
	$p_2 \gets a_n p_2$\\
	$p_3 \gets a_n p_3$\\
	$p_1 \gets p_A p_B a_n a_1 a_{n-1}$ \\
	\Return $p_1, p_2, p_3$ 
\end{algorithm}

\begin{theorem}\label{even_cycle}
	For $n \ge 6$, if $n$ is even then  $\mathcal{R}^p(C_n) = 3$.
\end{theorem}
\begin{proof}
	Suppose $C_n$ is an even cycle on $n \ge 6 $ vertices and  $\{a_1,a_2,\ldots,a_n\}$ is the vertex set of $C_n$ where $a_i$ is adjacent to $a_{i+1}$ for $1 \le i \le n-1$ and $a_1$ is adjacent to $a_n$. Note that $C_n \setminus \{a_n\} = P_{n-1}$, where $n-1$ is odd. By Remark \ref{unique}, the following $p_2$ and $p_3$ are the only permutations constituting a word with two permutations representing $P_{n-1}$ permutationally.
	\begin{align*}
		p_2 &= a_2 a_1 a_4 a_3 a_6 a_5 \cdots a_{n-2}a_{n-3}a_{n-1}\\
		p_3 &= a_{n-2}a_{n-1} a_{n-4} a_{n-3} \cdots a_6 a_7 a_4 a_5 a_2 a_3 a_1
	\end{align*}	

	As $n$ is even and the vertex $a_n$ is adjacent to $a_{1}$ and $a_{n-1}$, the vertex $a_n$ must precede $a_1$ and $a_{n-1}$ in permutations constituting a word representing $C_n$ permutationally. Accordingly, $p_2$ and $p_3$ can be updated to the following permutations:
	\begin{align*}
		p_2 &= a_n a_2  a_1 a_4 a_3 a_6 a_5 \cdots a_{n-2}a_{n-3}a_{n-1}\\
		p_3 &= a_n a_{n-2} a_{n-1} a_{n-4} a_{n-3} \cdots a_6 a_7 a_4 a_5 a_2 a_3 a_1 
	\end{align*}
	
	Note that $p_2p_3$ does not represent $C_n$ as the vertices $a_2,a_3, a_4, \ldots, a_{n-3},a_{n-2}$ are not adjacent to $a_n$, but they alternate in $p_2p_3$. 
	
	The permutation $p_1$ takes care of the non-adjacency of $a_n$ with the vertices $a_2,a_3, a_4, \ldots, a_{n-3},a_{n-2}$.	
	$$p_1 = p_A p_B a_n a_1 a_{n-1}$$	
	where $p_A$ and $p_B$ are permutations on the vertices of  $A \setminus \{a_n\}$ and $B \setminus \{a_1,a_{n-1}\}$, respectively. For $2 \le k \le n-2$, while $a_k a_n \le p_1$, we have $a_na_k \le p_2$ and $a_na_k \le  {p_3}$. Hence, $p_1p_2p_3$ represents $C_n$ permutationally.   \qed
	
\end{proof}

\section{Book Graphs}

The Cartesian product of two graphs $G = (V, E)$ and $H=(V', E')$, denoted by $G \square H $, is a graph with the vertex set $V \times V'$ and $(a,b), (c,d) \in V \times V'$ are adjacent if $a = c$ and $\overline{b d} \in E'$, or $b = d$ and $\overline{a c} \in E$. The Cartesian product of two word-representable graphs is word-representable (cf. \cite{kitaev15mono}). In \cite[Theorem 7]{broere_2019},  it was proved that  if $G$ is a word-representable graph with $\mathcal{R}(G)= k$, then $\mathcal{R}(G \square K_2) \le k+1$, where $K_2$ is the complete graph on two vertices. 

A book graph, denoted by $B_m$, is the Cartesian product of a complete bipartite graph $K_{1, m}$ (also called a star graph) and $K_2$. For $m \ge 2$, since $\mathcal{R}(K_{1,m})= 2$, we have $\mathcal{R}(B_m) \le 3$. Note that $B_1 = C_4$, and $B_2$ is the ladder graph on six vertices. It was shown in \cite{kitaev13} that the $\mathcal{R}(B_1) = \mathcal{R}^p(B_1) = \mathcal{R}(B_2) = \mathcal{R}^p(B_2) = 2$.

In the following, we show that the representation number of book graphs on at least eight vertices is three.

\begin{figure}
	\centering
	\begin{minipage}{.4\textwidth}
		\centering
		\begin{tikzpicture}[scale=0.85]
	\vertex (1) at (0,2.3) [label=above:$1$] {};  
	\vertex (2) at (-1.5,3) [label=above:$2$] {}; 
	\vertex (3) at (1,3) [label=above:$3$] {}; 
	\vertex (4) at (1.5,3) [label=above:$4$] {}; 
	\vertex (5) at (0,1.3) [label=below:$5$] {}; 
	\vertex (6) at (-1.5,0.7) [label=below:$6$] {}; 
	\vertex (7) at (1,0.7) [label=below:$7$] {}; 
	\vertex (8) at (1.5,0.7) [label=below:$8$] {}; 
	\path
	(1) edge (2)
	(1) edge (3)
	(1) edge (4)
	(1) edge (5)
	(2) edge (6)
	(3) edge (7)
	(4) edge (8)
	(1) edge (2)
	(1) edge (3)
	(1) edge (4)
	(1) edge (5)
	(2) edge (6)
	(3) edge (7)
	(4) edge (8)
	(5) edge (6)
	(5) edge(7)
	(5) edge (8);
	\end{tikzpicture}
	\caption{$B_3$}
	\label{B_3}
	\end{minipage}
	\begin{minipage}{.5\textwidth}
		\centering
	\begin{tikzpicture}[scale=0.7]
	\vertex (1) at (0,0) [label= left:$5 $] {};  
	\vertex (2) at (-2.2,-1.5) [label=below:$2 $] {}; 
	\vertex (3) at (2.2,-1.5) [label=below:$3 $] {}; 
	\vertex (4) at (0,1.8) [label=above:$4 $] {}; 
	
	\vertex (5) at (-0.5,0.3) [label= left:$ 8$] {}; 
	\vertex (6) at (0.7,0) [label=above:$7 $] {}; 
	\vertex (7) at (0,-0.7) [label=below:$ 6$] {}; 
	
	\vertex (0) at (3, -1) [label= right:$ 1$] {}; 
	\path
	(1) edge (2)
	(1) edge (3)
	(1) edge (4)
	(1) edge (5)
	(1) edge (6)
	(1) edge (7)
	(2) edge (7)
	(2) edge (3)
	(4) edge (2)
	(4) edge (5)
	(3) edge (6)
	(3) edge (4);
	\path[dashed]
	(4) edge (0)
	(3) edge (0)
	(2) edge (0)
	(1) edge (0);
	\end{tikzpicture}
	\caption{Local complementation of $B_3$ at 1}
	\label{X_1}
\end{minipage}

\end{figure}

\begin{lemma}
	The representation number of $B_3$ is 3.
\end{lemma} 
\begin{proof}
	Suppose $\mathcal{R}(B_3) \le 2$, i.e. $B_3$ is a circle graph. Since the circle graphs are closed under local complementation\footnote{The local complement of a graph at vertex $v$ is the complement of the induced subgraph on the vertex set $N(v)$.} (cf. \cite{bouchet}), consider the graph obtained by local complementation at vertex $1$ of $B_3$, as shown in Fig. \ref{X_1}. Note that its induced subgraph represented in thick edges in Fig. \ref{X_1} is not a word-representable graph (compare with the list of graphs in \cite[Figure 3.9]{kitaev15mono}).  Hence, the graph in Fig. \ref{X_1} is not a circle graph and so is $B_3$. Therefore, $\mathcal{R}(B_3) \ge 3$. Since $\mathcal{R}(B_3) \le 3$, we have $\mathcal{R}(B_3) = 3$. \qed
\end{proof}

Note that for all $m \ge 3$, $B_m$ contains $B_3$ as an induced subgraph so that $\mathcal{R}(B_m) \ge 3$. Since $\mathcal{R}(B_m) \le 3$, we have the following corollary.

\begin{corollary}
	For all $m \ge 3$, $\mathcal{R}(B_m) = 3$.
\end{corollary}

Note that a book graph is a bipartite graph and hence it is a comparability graph. Further, $ \mathcal{R}^p(B_m) \ge \mathcal{R}(B_m) = 3$. In the following, we use Algorithm \ref{algo-1}, to construct three permutations for a book graph and show that its \textit{prn} is three.

\begin{figure}[tbh!]
	\centering
	\begin{tikzpicture}[scale=0.7]
		\vertex (0) at (0,-1.5) [label=above:$0$] {};  
		\vertex (1) at (-2,-1) [label=above:$1$] {}; 
		\vertex (2) at (-1,-1) [label=above:$2$] {}; 
		\vertex (3) at (1,-1) [label=above:$3$] {}; 
		\vertex (4) at (3,-1) [label=above:$m$] {}; 
		\vertex (5) at (2,-1) [label=above:$ $] {}; 
		
		\vertex (10) at (0,-2.5) [label=below:$0'$] {};  
		\vertex (11) at (-2,-3) [label=below:$1'$] {}; 
		\vertex (12) at (-1,-3) [label=below:$2'$] {}; 
		\vertex (13) at (1,-3) [label=below:$3'$] {}; 
		\vertex (14) at (3,-3) [label=below:$m'$] {}; 
		\vertex (15) at (2,-3) [label=below:$ $] {}; 
		\path
		(0) edge (1)
		(0) edge (2)
		(0) edge (3)
		(0) edge (4) 
		(10) edge (11)
		(10) edge (12)
		(10) edge (13)
		(10) edge (14)
		(0) edge (5)
		(10) edge (15);
		
		\path[dashed]
		(1) edge (11)
		(2) edge (12)
		(3) edge (13)
		(4) edge (14)
		(0) edge (10)
		(15) edge (5);
		
	\end{tikzpicture}
	\caption{$B_m$}
	\label{Constb5}	
\end{figure}

Consider the book graph $B_m$ with vertex set $\{0, 1, \ldots, m\} \cup \{0', 1', \ldots, m'\}$ as shown in the Fig. \ref{Constb5}. Let $T$ be the star graph $K_{1,m}$ on the vertex set $\{0, 1, \ldots, m\}$ and $T'$ be the star graph $K_{1,m}$ on the vertex set $\{0', 1', \ldots, m'\}$.  Using Algorithm \ref{algo-1}, construct three permutation for $T$ with $0$ as the root. The three permutations are as follows:
\begin{align*}
	p_1 & = 1 2  \cdots (m-1) m 0 \\
	p_2 & = m (m-1) \cdots 2 1 0 \\
	p_3 & = 1 2  \cdots (m-1) m 0
\end{align*} 
Since $p_1p_2p_3$ represents $T$ permutationally, we have $r(p_1)r(p_2)r(p_3)$ also represents $T$ permutationally, where $r(p_i)$ is the reversal of the permutation $p_i$. Accordingly, consider the following permutations whose concatenation represents $T'$ permutationally.
\begin{align*}
	p_1' & = 0' m' (m-1)' \cdots 2' 1'\\
	p_2' & = 0' 1' 2' \cdots (m-1)' m' \\
	p_3' & = 0' m' (m-1)' \cdots 2' 1'
\end{align*} 
In addition to edges of $T$ and $T'$, we have the edges $\overline{ii'}$ (for $0 \le i \le m$) in $B_m$. Using $p_1,p_2,p_3,p_1',p_2'$ and $p_3'$, we construct the following three permutations in which $i$ and $j'$ alternate if and only if $i = j$:
\begin{align*}
  q_1 &= 0' 1 1' 2 2' \cdots (m-1)(m-1)' m m'0 \\
  q_2 &= 0' m m' (m-1)(m-1)' \cdots 2 2' 1 1'0\\
  q_3 &= 1 2 \cdots (m-1) m 0' 0 m' (m-1)' \cdots 2' 1'
\end{align*}
Here, the permutations $q_1$ and $q_2$ are constructed by shuffling the letters of $p_2'$ and $p_1$, and $p_1'$ and $p_2$, respectively. The permutation $q_3$ is constructed by concatenating $p_3$ and $p_3'$, and swapping $0$ and $0'$. It can be verified that $q_1q_2q_3$ represents $B_m$ permutationally. 
 
\begin{theorem}
 	For $m \ge 3$, $\mathcal{R}^p(B_m) = 3$.
\end{theorem}

\begin{remark}
	The book graphs is a class of graphs for which the representation number and \textit{prn} are the same.
\end{remark}
	
\section{Conclusion}

Modular decomposition of a graph $G$ is a partition of the vertex set of $G$ into modules, where a module $M$ is an induced subgraph of $G$ such that every vertex in $M$ has the same neighborhood outside $M$ \cite{Gallai}. In \cite{kitaev13}, it was proved that $\mathcal{R}(G') = \max \{\mathcal{R}(G), \mathcal{R}(M)\}$, where $G'$ is a modular extension of a word-representable graph $G$ by a module $M$ that is a comparability graph. Using the aforesaid result and the \textit{prn} of a tree, we may identify a large class of graphs with the \textit{prn} of at most three. For instance, if we replace every vertex of a tree with a module that is a permutation graph, then the resultant graph has the representation number two and the \textit{prn} at most three. Further, based on our observations, we state the following conjecture on the comparability graphs with \textit{prn} at most three. 

\begin{conjecture}
	Let $G$ be a comparability graph. If $G \in \mathcal{R}_2$ then $G \in \mathcal{R}^p_{\le 3}$.
\end{conjecture}

\appendix

\section{Demonstration of Algorithm \ref{algo-1}} \label{tree-demo}

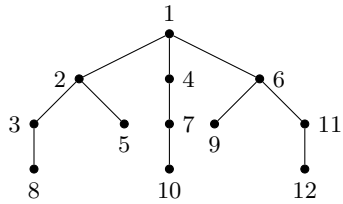
\begin{figure}
	\centering
	\begin{tikzpicture}[scale=0.6] 
	\vertex (1) at (0,0) [label=above:$1$] {};  
	\vertex (2) at (-2,-1) [label=left:$2$] {};
	\vertex (3) at (-3,-2) [label=left:$3$] {};  
	\vertex (4) at (0,-1) [label=right:$4$] {}; 	
	\vertex (5) at (-1,-2) [label=below:$5$] {}; 
	\vertex (6) at (2,-1) [label=right:$6$] {}; 
	\vertex (7) at (0,-2) [label=right:$7$] {}; 
	\vertex (8) at (-3,-3) [label=below:$8$] {}; 	
	\vertex (9) at (1,-2) [label=below:$9$] {}; 
	\vertex (10) at (0,-3) [label=below:$10$] {}; 
	\vertex (11) at (3,-2) [label=right:$11$] {}; 
	\vertex (12) at (3,-3) [label=below:$12$] {}; 
	
	\path
	(1) edge (2)
	(1) edge (4)
	(1) edge (6)
	(2) edge (3)
	(2) edge (5)
	(3) edge (8)
	(4) edge (7)
	(7) edge (10)
	(6) edge (9)
	(6) edge (11)
	(11) edge (12);
	\end{tikzpicture}
	\caption{A 3-ary tree $T$}
	\label{T}
\end{figure}

Consider the tree given in Fig. \ref{T}. Let $A = \{2,4,6,8,10,12\}$ and $B=\{1,3,5,7,9,11\}$. 
Firstly, as per the algorithm, $p_1 = p_2 = p_3 = 1$.
The vertices $2, 4, 6$ are the children of $1$. Since $1$ is odd, the permutations $p_1,p_2,p_3$ are updated as follows:
\begin{align*}
p_1 &= 2461\\
p_2 &= 6421\\
p_3 &= 2461
\end{align*}
Next, since $2<4<6$, the vertex 2 is visited first among the children of 1. Since $2$ is even and $3,5$ are the children of 2, the permutations $p_1,p_2,p_3$ are updated as follows:
\begin{align*}
p_1 &= 2 3 5 4 6 1\\
p_2 &= 6 4 2 1 5 3\\
p_3 &= 2 3 5 4 6 1
\end{align*}
Then, the vertex 4 is visited followed by the vertex 6. Note that when the vertex 6 is visited, the subword, $617$, is replaced by the word $6179(11)$ in $p_1$, where $D_4 = \{7\}$. The following are the  updated permutations:
\begin{align*}
p_1 &= 2 3 5 4 6 1 7 9 (11) \\
p_2 &= 6 (11) 9 4 7 2 1 5 3\\
p_3 &= 2 3 5 4 7 6 9 (11) 1
\end{align*}
Next, the vertex 3 is visited, followed by 7, 9 and 11. Note that 5 is a leaf, so it will not be visited. The algorithm ends as the vertices 8, 10, and 12 are leaves of $T$, and $Q$ becomes empty. The following are the final permutations: 
\begin{align*}
p_1 &= 2 8 3 5 4 6 1 (10) 7 9 (12) (11)  \\
p_2 &= 6 (12) (11) 9 4 (10) 7 2 1 5 8 3 \\
p_3 &=  (12) (10) 8 2 3 5 4  7 6 9 (11) 1
\end{align*}

\section{Details of the Proof of Lemma \ref{path}}
\label{app_lem4}

\begin{itemize}
	\item Case 1: $n$ is odd. As per Algorithm \ref{algo-2}, it is evident that 			
	$$p_2 = a_2 a_1 a_4 a_3 a_6 a_5\cdots a_{n-1}a_{n-2}a_n$$
	$$p_3 = a_{n-1} a_n a_{n-3} a_{n-2} \cdots a_6 a_7 a_4 a_5 a_2 a_3 a_1$$
	
	Suppose $a_j \in P_n$.  
	\begin{itemize}
		\item[-] If $j = 1$, note that $a_1$ is adjacent to only $a_2$ in $P_n$. Clearly, $a_2 a_1 \le p_2$ and $p_3$ both. Further, for $3 \le k \le n$, $a_1 a_k \le p_2$. Whereas, $a_k a_1 \le p_3$. Hence, we are through when $j = 1$.
		
		\item[-] Similarly, if $j = n$, note that $a_n$ is adjacent to only $a_{n-1}$ in $P_n$ and $a_{n-1}a_n \le p_2$ as well as $p_3$. On the other hand, for $1 \le k \le n-2$, $a_k a_n \le p_2$ and $a_n a_k \le p_3$, as desired. 
		
		\item[-] For $2 \le j \le n-1$, note that $a_j$ is adjacent only to $a_{j-1}$ and $a_{j+1}$. We consider the following two subcases depending on whether $j$ is odd or even: 
		\begin{itemize}
			\item[-] Suppose $j$ is odd. The vertex $a_j$ appears in $p_2$ and $p_3$ as shown below: 
			$$p_2= a_2 a_1 a_4 \cdots a_{j-1} a_{j-2} a_{j+1}a_j a_{j+3}a_{j+2} \cdots a_{n-1} a_{n-2} a_{n} $$
			$$ p_3 = a_{n-1} a_{n} a_{n-3} a_{n-2} \cdots  a_{j+1} a_{j+2} a_{j-1} a_j a_{j-3} a_{j-2} \cdots a_2 a_3 a_1$$
			
			Clearly, the pairs of vertices $\{a_{j-1}, a_{j}\}$ and $\{a_{j+1}, a_{j}\}$ alternate in both $p_2$ and $p_3$. Whereas, when $ j+2 \le k \le n$, $a_j a_k \le p_2$ and $a_k a_j \le p_3$. Also, when $ 1 \le k \le j-2$, $a_k a_j \le p_2$ and $a_j a_k \le p_3$.
			
			\item[-] Suppose $j$ is even. The vertex $a_j$ appears in $p_2$ and $p_3$ as shown below:  
			$$p_2= a_2 a_1 a_4 \cdots a_{j-2} a_{j-3}a_j a_{j-1}a_{j+2}a_{j+1} \cdots a_{n-1} a_{n-2} a_{n} $$
			$$ p_3 = a_{n-1} a_{n} a_{n-3} a_{n-2} \cdots  a_{j+2} a_{j+3} a_j a_{j+1} a_{j-2} a_{j-1} \cdots a_2 a_3 a_1$$
			
			Clearly, $a_{j}a_{j-1} \le p_2$ and $ a_j  a_{j-1} \le p_3$. So also, $a_{j}a_{j+1} \le p_2$ and $ a_j a_{j+1}  \le p_3$.  Whereas, $a_j a_k \le p_2$ and $a_k a_j \le p_3$, where $ j+2 \le k \le n$. Also, 	$a_k a_j \le p_2$ and $a_j a_k \le p_3$, where $ 1 \le k \le j-2$.
		\end{itemize}
	\end{itemize}
	\item Case 2: $n$ is even.	As per Algorithm \ref{algo-2}, it is clear that				
	$$p_2 = a_2 a_1 a_4 a_3a_6 a_5 \ldots a_{n-2}a_{n-3}a_n a_{n-1}$$
	$$p_3 = a_n a_{n-2} a_{n-1} a_{n-4} a_{n-3} \ldots a_6 a_7 a_4 a_5 a_2 a_3 a_1$$
	As in Case 1, in the following, we show that two vertices of $P_n$ are adjacent  if and only if they alternate in $p_2p_3$. 
	
	Suppose $a_j \in P_n$.  
	\begin{itemize}
		\item[-] If $j = 1$, note that $a_1$ is adjacent to only $a_2$ in $P_n$. Clearly, $a_2 a_1 \le p_2$ and $p_3$ both. Further, for $3 \le k \le n$, $a_1 a_k \le p_2$. Whereas, $a_k a_1 \le p_3$. Hence, we are through when $j = 1$.
		
		\item[-] Similarly, if $j = n$, note that $a_n$ is adjacent to only $a_{n-1}$ in $P_n$ and $a_n a_{n-1} \le p_2$ as well as $p_3$. On the other hand, for $1 \le k \le n-2$, $a_k a_n \le p_2$ and $a_n a_k \le p_3$, as desired. 
		
		\item[-] For $2 \le j \le n-1$, note that $a_j$ is adjacent only to $a_{j-1}$ and $a_{j+1}$. We consider the following two subcases depending on whether $j$ is odd or even: 
		\begin{itemize}
			\item[-] Suppose $j$ is odd. The vertex $a_j$ appears in $p_1$ and $p_3$ as shown below: 
			$$p_2= a_2 a_1 a_4 \cdots a_{j-1} a_{j-2} a_{j+1}a_j a_{j+3}a_{j+2} \cdots  a_{n-2} a_{n-3}a_{n} a_{n-1} $$
			$$ p_3 = a_{n} a_{n-2} a_{n-1} a_{n-4} \cdots  a_{j+1} a_{j+2} a_{j-1} a_j a_{j-3} a_{j-2} \cdots a_2 a_3 a_1$$.
			
			Clearly, the pairs of vertices $\{a_{j-1}, a_{j}\}$ and $\{a_{j+1}, a_{j}\}$ alternate in both $p_2$ and $p_3$. Whereas, when $ j+2 \le k \le n$, $a_j a_k \le p_2$ and $a_k a_j \le p_3$. Also, when $ 1 \le k \le j-2$, $a_k a_j \le p_2$ and $a_j a_k \le p_3$.
			
			\item[-] Suppose $j$ is even. The vertex $a_j$ appears in $p_2$ and $p_3$ as shown below:  
			$$p_2= a_2 a_1 a_4 \cdots a_{j-2} a_{j-3}a_j a_{j-1}a_{j+2}a_{j+1} \cdots a_{n-2}a_{n-3} a_{n} a_{n-1} $$
			$$ p_3 = a_{n} a_{n-2} a_{n-1}a_{n-4}a_{n-3} \cdots  a_{j+2} a_{j+3} a_j a_{j+1} a_{j-2} a_{j-1} \cdots a_2 a_3 a_1$$
			
			Clearly, $a_{j}a_{j-1} \le p_2$ and $ a_j  a_{j-1} \le p_3$. So also, $a_{j}a_{j+1} \le p_2$ and $ a_j a_{j+1}  \le p_3$.  Whereas, $a_j a_k \le p_2$ and $a_k a_j \le p_3$, where $ j+2 \le k \le n$. Also, 	$a_k a_j \le p_2$ and $a_j a_k \le p_3$, where $ 1 \le k \le j-2$.
		\end{itemize}
	\end{itemize}
\end{itemize}

\end{document}